\documentclass[oneside,11pt]{amsart}

\usepackage[headinclude]{typearea}
\usepackage{amsmath,amsthm,amsfonts,amssymb}
\usepackage{hyperref}
\usepackage{ifpdf}
\usepackage{a4}
\usepackage{color}
\usepackage{graphicx}
\usepackage{epsfig}
\usepackage{caption}

\theoremstyle{plain}
\newtheorem{theorem}{Theorem}[section]
\newtheorem{lemma}[theorem]{Lemma}

\theoremstyle{definition}
\newtheorem{definition}[theorem]{Definition}
\newtheorem{remark}[theorem]{Remark}
\newtheorem{assumption}[theorem]{Assumption}

\newcommand{\B}{{\mathbb{B}}}
\newcommand{\E}{{\mathbb{E}}}

\renewcommand{\P}{{\mathbb{P}}}

\newcommand{\R}{{\mathbb{R}}}

\makeatletter
\DeclareRobustCommand*\cal{\@fontswitch\relax\mathcal}
\makeatother

\renewcommand{\P}{{\mathbb P}}

\newcommand{\cA}{{\cal A}}

\newcommand{\cG}{{\cal G}}

\newcommand{\cP}{{\cal P}}
\newcommand{\cR}{{\cal R}}

\newcommand{\be}{\begin{equation}}
\newcommand{\ee}{\end{equation}}
\newcommand{\bea}{\begin{eqnarray}}
\newcommand{\eea}{\end{eqnarray}}
\newcommand{\beast}{\begin{eqnarray*}}
\newcommand{\eeast}{\end{eqnarray*}}
\newcommand{\bproof}{\begin{proof}}
\newcommand{\eproof}{\end{proof}}

\title[Pricing of derivatives written on industrial loss indexes]{Utility indifference pricing of derivatives written on industrial loss indexes}
\thanks{The first author is partly supported through the Austrian Science Fund
project P21196. The second author is supported by the Austrian Exchange Service
under the North-South-Dialogue Scholarship programme.}
\author{Gunther Leobacher and  Philip Ngare}

\begin{document} 

\bibliographystyle{alpha}

%\date{\today}
\maketitle

\begin{abstract}  We consider the problem of pricing derivatives 
written on some industrial loss index via utility
indifference pricing. The industrial loss index is modeled by a compound
Poisson process and the insurer can adjust her portfolio by choosing the risk
loading, which in turn determines the demand. We compute the price of a CAT
(spread) option written on that index using utility indifference pricing.\\
\end{abstract}

{\noindent \bf Keywords:} {(Re-)Insurance, catastrophe derivatives, jump process, random thinning, utility  indifference price\\[-0.5em]}

{\noindent \bf MSC2010:}  {91B16, 91G20, 93E20, 60J75 \\[-0.5em]}


%\tableofcontents

\section{Introduction}\label{sec:introduction}

%Natural forces such as earthquakes, hurricanes and landslides often leave
%human and economic losses in their wake. Such hazards are considered natural
%disasters when they lead to extremely large losses, which is the case when
%they affect densely populated areas. They are not very frequent, but their
%effect on economic life can be devastating. However, many developing
%countries take few precautions to lessen the impact of disasters and local
%insurance markets are unable to satisfy the risk financing requirements. 
%Derivatives linked to catastrophic events  can be used 
%as  the major security  in the alternative risk transfer. \\

It was recognized shortly after the Hurricane Andrew in 1992, then the most
costly natural catastrophe in history, that 
events of this magnitude significantly stress the capacity of the insurance
industry. On the other hand, the accumulated losses of those events are rather
small relative to the US stock and bond markets. Thus, securitization offers
a potentially more efficient mechanism for financing CAT losses than
conventional insurance and reinsurance, see Cummins et al, \cite{clp}.\\

The first contracts were launched by the Chicago Board of Trade (CBOT), which
introduced catastrophe futures in 1992 and later introduced catastrophe put and
call options.  The options were based on aggregate catastrophe-loss indices
compiled by Property Claims Services, an insurance industry statistical agent,
see \cite{cum}. 

In the absence of a traded underlying asset, insurance-linked securities have
been structured to pay-off on three types of variables: Insurance-industry
catastrophe loss indices, insurer-specific catastrophe losses, and parametric
indices based on the physical characteristics of catastrophic events. 
The first variant involves higher basis risk and less exposure to moral
hazard than the second, the third variant tries to balance the two risks
in a suitable way, cf. Cummins \cite{cum}. In this paper we solely concentrate on 
index-based derivatives.\\

A simple example of such a derivative is provided by the aforementioned
call options on an insurance-industry catastrophe loss index. The variant 
introduced by CBOT was actually a call option spread, that is, a combination
of a call option long and another call option short with a higher strike. 

A more popular type of catastrophe derivative is the CAT bond. This is a 
classical bond in which there is an option embedded which is triggered
by a defined catastrophic event. In this paper we will again only
consider those bonds where this catastrophic event depends on some 
industry-loss index, though in practice both of the other variants are
of importance as well. From our point of view there is little difference
between CAT bond and CAT option, since on evaluating a CAT bond we 
concentrate on the embedded option. There is however some danger of confusion
regarding the role of buyer/seller with CAT bonds: The issuer of the bond
actually buys the embedded option while the buyer of the bond sells the
option. \\ 

For the issuer of a CAT bond -- typically
an insurance or reinsurance company -- it serves as a reinsurance. 
On the other hand, the investor who buys the bond (and therefore sells
an option) receives 
a coupon over the market interest and can, at the same time, diversify her
risk by investing in a security whose payoff is largely uncorrelated with
classical financial instruments. \\
 
%The  arrival rate of catastrophic events can be described by a pure 
%Poisson process.
Geman and Yor \cite{Geman} analyze catastrophe
options with payoff $(C(T)-K)^+$ where $C$ is the aggregate claims process 
which is
modeled by a jump-diffusion process. Cox \cite{Cox} used a  pure Poisson
process to model the aggregate loss of an insurance company, and derived the
pricing formula of CATEputs under the assumptions of constant arrival rates of
catastrophic events. Jaimungal and Wang \cite{Jaimungal} used a
compound Poisson process to describe the dynamic losses more accurately, but
maintain the assumption of the constant arrival rate of claims.\\

%Like in \cite{Jaimungal}, our market is incomplete.  
We  model the arrival of
claims, which are accounted for in some industrial loss index, as a  Poisson
process with fixed arrival intensity.  The underlying of the CAT derivative,
the index, is itself not tradable.  It therefore makes sense to use the method
of indifference pricing via expected utility of Hodges and Neuberger \cite{Hodges} to price the
derivative. A similar approach can be found in Egami and Young  \cite{Egami}, where the
authors used utility indifference pricing techniques to price structured
catastrophe bonds.  However, there is a big difference in our modeling of the
hedging opportunity.  In our setup this is done via adjusting the insured
portfolio.\\

For catastrophic events, the assumption that the resulting claims
occur at jump times of a Poisson process as adopted by most previous studies is
not beyond justifiable critique. 
Therefore 
alternative point processes have  been used to generate the claim arrival
process. Lin et al \cite{Lin} proposed a doubly stochastic Poisson
process, (also called ``Cox process'', see 
\cite{Cox,Grandell1,Grandell2,Bremaud,Lando}) 
to model the arrival process for catastrophic events
and derived pricing formulas of contingent capital. See also Fuita et al \cite{fuji}
for arbitrage pricing of  CAT bonds in such a context. 
Jaimungal and Chong \cite{Jaimungal2} consider valuation of catastrophe 
derivatives when the rate of the claims is modulated by a Markov chain.
\\

Charpentier \cite{charpentier} considers hedging of catastrophe derivatives
with stocks whose jumps depend on catastrophic events and how to compute
a utility indifference price in this setup. 

Our study contributes to the literature  by presenting a new approach
to hedging a CAT derivative via adjustment of the insured portfolio, which
in turn is done via adjusting the risk loading and an exogenously given
demand curve. The main idea is that the
loss in the portfolio of a single insurance company is necessarily
correlated with an industrial loss index that includes the losses
of that insurance. The introduction of the derivative has therefore an influence
on the pricing policy of the insurance company.

It has been noted by  Cummins \cite{cum} that the relatively low volume in the
CAT derivatives market may in part be due to insufficient understanding
of how these products may be hedged. Our paper gives a new perspective
to the hedging of CAT derivatives via the most basic operation of an insurance
company, i.e. the choice of a suitable risk loading for a particular risk.
Future work may combine this approach with other hedging methods, like
trading in shares that are correlated with catastrophic events, such as 
those of construction companies.\\

The paper is organized as follows: In Section \ref{sec:model-setup} we give
the problem description: We assume a global claims process $C$, which keeps
track of all claims due to a specific event in a given country and we consider
an insurance company in the same country, so that the index will contain  the
losses of that particular insurance company among others.  The insurance
company is  facing a certain demand curve which determines the fraction of the
insurance market that the company gets to insure, dependent on the risk loading
it charges.  We therefore have to model the $\xi$-fraction of the claims
process $C$, for an insurance company  with a $\xi$-fraction of the market.
Such a model is constructed in Section \ref{sec:model-setup} where  we also
derive the wealth process for the insurance company. We conclude that section 
by giving a short introduction into the concept of utility indifference
pricing.

Section \ref{sec:computation} constitutes the main part of our paper: 
We  derive  a suitable Hamilton-Jacobi-Bellman (HJB) equation for an insurance 
that holds
$k$ units of a derivative written on the total number of claims
(Subsection \ref{sec:optimal-strategy}). We use the concept of 
{\em piecewise deterministic Markov decision process} as presented,
e.g., by B\"auerle and Rieder in \cite{BR2010,BR2011}. In particular,
we will make use of a verification theorem from \cite{BR2011} to show
that a solution to the HJB equation also solves the optimal control problem.
At that stage we will have to specialize to exponential utility.  

Two subsections, \ref{sec:linear-demand} and \ref{sec:non-linear_demand},
are devoted to special demand functions. While Subsection 
\ref{sec:linear-demand} looks into the details of the very special case
of linear demand, Subsection  \ref{sec:non-linear_demand} presents a 
class of demands that are more general than the previous linear one, but 
still preserve the property of leading to a unique optimal risk loading.
By the end of Section \ref{sec:pricing}, in Subsection
\ref{sec:numerical-example}, we give a numerical example for linear demand.
In the technical Subsection \ref{sec:verification} we show that the conditions 
of the verification theorem are satisfied. 

Section \ref{sec:limit-prices} is devoted to the question under which conditions
the derivative could actually be sold, that is, when the buyer's price is at
least as big as the seller's price. To that end we study a couple of 
different pricing concepts related to the utility indifference price.

\section{Model setup}
\label{sec:model-setup}

\subsection{Study problem}

Suppose we have a global claims process $C=(C_t)_{t\ge 0}$, which keeps track
of all insurance claims due to a specific type of event in a country. 
That is, $C_t$ is the cumulative sum of all claims up to time $t$.
We assume that there are
$M$ possible clients in the market which potentially contribute to the claims
process. If all those clients had insurance 
contracts with the same insurance company, then $C$ were the claims 
process of this insurance company. Let $a$ be the ``fair'' annual premium 
for one client, that is $\E(C_1)=M\cdot a$.
The annual premium for one contract
therefore has to be
greater or equal than $a$, since otherwise the insurance will make an
almost sure loss in the long run. 

Assume that an insurance company faces a demand curve $q$ such that if 
the premium the insurance charges for the claim is $a(1+\theta)$, then 
the company gets to insure the $\frac{q(\theta)}{M}$-th part of the whole claim 
process for the total annual premium $a(1+\theta) q(\theta)$, where $M$ 
is the total number of clients. 
It is assumed that
$q$ is continuous and strictly decreasing in $\theta$. 
We further make the reasonable assumption that the insurance gets 
to insure the whole 
process if it does not charge any risk loading (any strictly risk-averse
client would enter such a contract) and it gets $0$ contracts if
the risk-loading exceeds some fixed number $m>0$. In our model $\theta$
may vary over time and we assume that the number of contracts is adjusted
instantaneously via the demand function $q$. This is a simplifying assumption
that will not be met in practice. However we will see in our numerical
examples that there are large areas where the risk loading is almost constant
over time and the index value. Thus our analysis can either be viewed as 
a first order approximation to more realistic models or as a model that has
practical relevance only under certain market conditions.

Note that in the above setting the insurance company is 
not necessarily a monopolist: $q$ is the demand that the company faces,
which may well be influenced by other insurance companies' decisions.
We only assume that competing firms respond consistently to changes 
in the risk loading, so that $q$ does not change over time and is not 
influenced by earlier policy decisions.
We will ignore the influence of the derivative that we want to price
 on demand, i.e. we 
assume that the demand the insurance company faces on its insurance contracts
is the same regardless
of whether the insurance company introduces the derivative or not.

One question arising here is how we can
model the $\xi$-th part of the industry loss process in a way that 
an insurance company which holds contracts for the $\xi$-th part of 
the market will only be confronted with the $\xi$-th part of
the claims? For fixed $\xi$ this will be a thinning of the original process.
Once we have found a model for this, we find that the wealth process of
the insurance company can 
be controlled via the risk loading process and therefore we can ask for
optimal strategies for maximizing terminal utility.
This in turn will make it possible to use the method of utility indifference
pricing for CAT-derivatives.

Since the wealth process is obviously correlated to the global claims process
$C$, any derivative written on $C_T$, for some fixed $T>0$, can be partially
hedged. This will result in a utility indifference price that is 
different from the utility equivalence price.

%We will use this to find a utility indifference price for such
%a derivative.

%We have to remark here that in the model presented in 
%Section \ref{sec:model-setup} the global claims arrive at the jump times of a
%Poisson process with constant intensity, which means that strictly speaking
%there are no real catastrophes happening. So if, for example, we consider a
% CAT-option, then it
%just compensates for an unusually high number of claims. Future extensions
%of the model should improve upon the present one by
%allowing for time spans with higher intensity of claim arrival, for example
%by letting the claims arrive at the jump times of a doubly stochastic
%Poisson process. 

\subsection{Exposure to industry loss}

We assume as given a Poisson process $N$ with intensity $\lambda M$, which 
models the arrivals of claims, as well
as sequences of i.i.d. random variables $Y_1,Y_2,\ldots$ with values in $\R_+$,
the sizes of the claims. 
Moreover we assume that there are i.i.d. random variables 
$U_1,U_2,\ldots$ with values in the space of all possible insurance clients.
The $Y_k$'s and the $U_k$'s are assumed to be independent of each other and 
independent  of $N$.
That is, $N$ tells us at which time $\tau_k$ the $k$-th claim occurs, 
$Y_k$ models its
size and $U_k$ tells us who is affected. Actually, from the point of view of
an insurance company the only interesting information about $U_k$ is whether
it is one of their own clients who is affected or not. 
We will therefore assume that the $U_k$'s are uniformly distributed 
on $[0,1]$ and
that a particular insurance company which 
holds the $\xi$-th part of the market is therefore affected with 
probability $\xi$. 

Hence the claims process for this insurance company 
can be modeled by
$$
C_t^\xi:=\sum_{k=1}^{N_t}Y_k 1_{U_k\le\xi_{\tau_k}}\,.
$$
The cumulative claims process of all claims constitutes our industrial
loss index and is defined as
$$
C_t:=C^1_t=\sum_{k=1}^{N_t}Y_k \,.
$$
Note that, for constant $\xi\in[0,1]$, the process 
$C^\xi$ is a {\em thinning} of $C$ and therefore is a compound Poisson process 
with intensity $\xi\lambda M$ (see, e.g., \cite[Section 4.4]{Resnick}). 
The joint process $(C,C^\xi)$ is therefore a compound Poisson process with 
values in $[0,\infty)\times[0,\infty)$ and with jump distribution
\[
\P((\Delta C_{\tau_k},\Delta C^\xi_{\tau_k})\in A)=\xi \P((Y_k,Y_k)\in A)+
(1-\xi)\P((Y_k,0)\in A)\,.
\]
for all Borel measurable sets $A\in [0,\infty)\times[0,\infty)$.\\

\subsection{Demand function and wealth process}

We now turn to the function $q$ which determines the fraction of
the market that the company gets to insure depending on the risk
loading it chooses.
We assume that $q$ is a strictly decreasing continuous 
function on 
$\R$ with $q(\theta)=M$ for $\theta\le 0$ and $q(\theta)=0$ for 
$\theta\ge m$, where $m$ is some positive real number. Therefore 
$q$ is rather general, even the requirement that it vanishes above
some level $m$ is rather innocuous: suppose the annual premium were much 
larger than the expected claim size, then surely this insurance could not be
sold. \\

We further define $a:=\lambda\E(Y_1)$, the expected
annual claim per client. In particular we assume $\E(Y_1)<\infty$.
In fact we will have to assume the stronger assumption $\E(e^{b Y_1})<\infty$
for some given constant $b>0$. 

For a measurable function $\theta:[0,\infty)\rightarrow\R$ 
we therefore define the dynamics of the 
wealth process for the insurance company
as
\begin{eqnarray}
X^\theta_t
&:=&x_0+\int_0^ta(1+\theta_s)q(\theta_s)ds-\sum_{k=1}^{N_t}Y_k 1_{U_k\le q(\theta_{\tau_k})/M}\,, \label{eq:Xdyn}\\
&=&x_0+\int_0^ta(1+\theta_s)q(\theta_s)ds-C^{q(\theta)/M}_t
\end{eqnarray}
where $q$ is the absolute demand function 
as introduced above, and $x_0$ is the initial wealth.

The process $(C,X^\theta)$ is a special case of a 
{\em piecewise deterministic Markov process}, where the flow does not depend
on $X$. See \cite[Chapter 8]{BR2011} or \cite{BR2010} for the definition and
theory of piecewise deterministic Markov processes. The corresponding data 
is given by
\begin{itemize}
\item the state space $[0,\infty)\times \R$;
\item the control space $\R$ (we will later see that we may restrict the
controls to the compact space $[0,m]$);
\item the deterministic flow $d(C_t,X^\theta_t)=(0,a q(\theta_t)(1+\theta_t)) dt$
between jumps;
\item the jump intensity $M\lambda$;
\item the stochastic kernel $Q$, 
\begin{align*}
Q(A|(c,x),\theta)&=
q(\theta)\P((Y,-Y)\in A-(c,x))\\
&\quad+(1-q(\theta))\P((Y,0)\in A-(c,x))\,,
\end{align*}
where $A-(c,x))=\{(a_1-c,a_2-x):(a_1,a_2\in A)\}$;
\item the zero reward rate;
\item the discount rate, which we set 0 for simplicity. 
\end{itemize}

Here and throughout the paper, $Y$ denotes a random variable having the
same distribution as the $Y_k$.

A Markovian control for this system is a measurable function
\[
f:[0,\infty)\times\R\times [0,\infty)\longrightarrow \{\theta:[0,\infty)\rightarrow \R, \theta \text{ measurable}\}
\]
which describes for input data $(C_\tau,X_\tau,\tau$) ---
where $\tau$ is the time of
a jump and $(C_\tau,X_\tau)$ is the state of the process immediately after the jump --- 
the control until the next jump.

\subsection{Utility and utility indifference pricing}\label{sec:pricing}

Throughout the paper we assume that investors and insurance companies 
have a utility function, by which we
mean a function $u:\R\longrightarrow\R$ which is strictly increasing
and concave, and that they aim to maximize the expected
utility of their wealth at some time $T$ in the future.\\

We recall the general idea of utility indifference pricing as
introduced in Hodges and Neuberger \cite{Hodges}.
An excellent introduction is Henderson and Hobson \cite{henhob}, 
from which we will repeat the basic definitions. 

The utility indifference buy (or bid) price $p^b$ is the price at which the
investor is indifferent (in the sense that his expected utility under 
optimal trading is unchanged) between paying nothing and not having the claim
$Z_T$ and paying $p^b$ now and receiving the claim $Z_T$ at time $T$.

Consider the problem with $k\ge 0$ units of the claim. Assume an investor with 
utility $u$ who initially
has wealth $x$ and zero endowment. Define 
$$
V(x,k):=\sup_{X_T\in\cA(x)}\E(u(X_T+k Z_T))
$$
where $\cA(x)$ is the set of all wealths $X_T$ which can be generated 
from initial fortune $x$ by following admissible strategies. 
The {\em utility indifference buy price} 
$p^b(k)$ is the solution to 
\be\label{eq:buy-price-def}
V(x-p^b(k),k)=V(x,0)\,.
\ee
That is, the investor is willing to pay at most the amount $p^b(k)$ today for
$k$ units of the claim $Z_T$ at time $T$. Similarly the 
{\em utility indifference sell price} 
$p^s(k)$ is the smallest amount the investor is willing to accept in order
to sell $k\ge 0$ units of $Z_T$. That is, $p^s(k)$ solves
$$
V(x+p^s(k),-k)=V(x,0)\,.
$$
The two prices are related via $p^b(k)=-p^s(-k)$.
With this in mind we can define the {\em utility indifference price}
$p(k)$ as the solution to (\ref{eq:buy-price-def}) for all $k\in \R$.

Henderson and Hobson \cite{henhob} note two features of the utility indifference price.
\begin{itemize}
\item {\em Non-linear pricing:} In contrast to the Black-Scholes price (and
many alternative pricing methodologies in incomplete markets), utility 
prices are generally non-linear in the number of claims, i.e. $k$.
\item {\em Recovery of complete market price:} If the market is complete or
if the claim $Z_T$ is replicable, the utility indifference price $p(k)$
is equal to the complete market price for $k$ units of the claim. 
\end{itemize}

It should be noted that, in general, the utility indifference price $p(k)$ 
also depends
on $x$. 
This dependence usually vanishes for exponential utility:
 Suppose that  $u(x)=-\exp(-\eta x)$, for some $\eta>0$.
Then
\begin{align*}
V(x-p(k),k)&=\sup_{\theta\in\Theta}
\E\left(u(x-p(k)+G^\theta_T+kZ_T)\right)\\
&=-e^{-\eta x}e^{\eta p(k)}\inf_{\theta\in\Theta}\E\left(\exp(-\eta (G^{\theta}_T+kZ_T))\right)\,,\\
V(x,0)
&=-e^{-\eta x}\inf_{\theta\in\Theta}\E\left(\exp(-\eta (G^{\theta}_T))\right)\,,
\end{align*}
such that, with our former notation
\be
p(k)=-\frac{1}{\eta}\left(\log\left(\inf_{\theta\in\Theta}\E(\exp(-\eta (G^{\theta}_T+kZ_T)))\right)
-\log\left(\inf_{\theta\in\Theta}\E(\exp(-\eta G^{\theta}_T))\right)\right)\label{eq:buyers-price}
\ee
(provided that the arguments in the logarithms are finite).

See also \cite{becherer} for additional properties of the utility 
indifference price for exponential utility.

We mention another feature of the utility indifference price: Suppose
the payment $Z_T$ is independent of $G^\theta_T$ for every choice of
$\theta$. 
Then 
$$
\E(\exp(-\eta (G^{\theta}_T+kZ_T)))
=\E(\exp(-\eta G^{\theta}_T)\E(\exp(-\eta kZ_T))
$$
and therefore 
$$
p(k)=-\frac{1}{\eta}\log\E(\exp(-\eta kZ_T))\,.
$$
This price is also called the {\em certainty equivalence price} of the 
derivative $Z_T$.  Note that one special case where $Z_T$ is independent 
of $G^\theta_T$ occurs when $G^\theta_T$ is deterministic.

\section{Computation of the utility indifference price}\label{sec:computation}

\subsection{Optimal dynamic risk loading}\label{sec:optimal-strategy}

We now want to apply the concept of utility indifference pricing 
to the model presented in Section \ref{sec:model-setup}.
Consider a derivative written on $C_T$, the total claims process
at time $T$.
Let its payoff be of the form $\psi(C_T)$ where $\psi$ is 
a continuous and bounded function on $[0,\infty)$. For example, 
if the derivative is 
a CAT (spread) option then
$\psi$ has the form
$$
\psi(c)=\max(0,\min(c-K,L-K))\,.
$$

For a given utility function $u$ we want to maximize expected utility 
from terminal wealth, i.e.  we want to compute
$$
\sup_{\theta}\E(u(X^\theta_T+k\psi(C_T)))\,,
$$
where $\theta$ ranges over all Markovian controls.

We have the following simple lemma which allows us to limit our 
considerations to bounded $\theta$:

\begin{lemma}\label{th:theta-g-0}
Let $\theta:[0,\infty)\rightarrow \R$ be measurable. Define 
$\kappa$ by
$$
\kappa=\left\{
\begin{array}{cccc}
0&\mbox{ if }&\theta<0\\
\theta& \mbox{ if }&0\le\theta\le m\\
m & \mbox{ if }&\theta > m\,.
\end{array}
\right.
$$
Then $\kappa$ is measurable with values in $[0,m]$ and for all $\tau_n\le t<\tau_{n+1}$ 
\[
X^\kappa_t-X^\kappa_{\tau_n}\ge X^\theta_t-X^\theta_{\tau_n}\,.
\]
\end{lemma}

\begin{proof}%[Proof of Lemma \ref{th:theta-g-0}]
It holds that $q(\kappa)=q(\theta)$ throughout by our assumptions on $q$. 
We therefore have
\begin{align*}
X^\kappa_t-X^\kappa_{\tau_n}-X^\theta_t+X^\theta_{\tau_n}
&=\int_{\tau_n}^t a(1+\kappa_s)q(\kappa_s)- a(1+\theta_s)q(\theta_s)ds\\
&=a \int_{\tau_n}^t (\kappa_s-\theta_s) q(\kappa_s)\,\ge 0\,,
\end{align*}
since $\kappa\ge \theta$ if $\theta\le m$ and $q(\kappa)=0$ if 
$\theta> m$.
\eproof

Define the value function of the problem as 
\be\label{eq:value_f}
V(t,x,c,k):=\sup_{\theta}\E(u(X^\theta_T+k\psi(C_T))|X^\theta_t=x,C_t=c)\,,
\ee
where the supremum is taken over all Markovian controls $\theta$. 
According
to Lemma \ref{th:theta-g-0} we may concentrate on $\theta$ with 
values in $[0,m]$. 
Note that due to the boundedness of $\psi$ we have
for $\theta\equiv m$ that
$\E(u(X^\theta_T+k\psi{(C_T)})|X^\theta_t=x,C_t=c)>-\infty$ for all $x,c,k$,
such that $V(t,x,c,k)>-\infty$.
Furthermore the expected terminal wealth $X_T$ is bounded from  above,
since the growth rate of $X$ is bounded and only negative jumps can occur. 
$V$ is therefore a well-defined real-valued function.

\cite{BR2010} prove that under fairly general conditions there exists an
optimal {\em relaxed} control for finite horizon problems 
for piecewise deterministic Markov decision processes. They also give 
conditions under which there exists a {\em nonrelaxed} policy. Those
later conditions are not satisfied for our problem, but relaxed controls are
to weak for the purpose of utility indifference pricing.
We therefore take a slightly different path using a Hamilton-Jacobi-Bellman 
(HJB) equation and a slight variation of the verification theorem 
\cite[Theorem 8.2.8]{BR2011}. \\

\begin{definition}
A measurable function $b:[0,\infty)\times\R\longrightarrow [0,\infty)$ is called a
{\em bounding function} for the piecewise deterministic Markov decision model,
if there exist constants $c_u,c_Q,c_{\text{flow}}\ge 0$ such that
for all $c\in [0,\infty),x\in \R$
\begin{enumerate}
\renewcommand{\theenumi}{\roman{enumi}}
\item $|u(c,x)|\le c_u b(c,x)$;
\item $\int b(c,x)Q(dc\times dx|c,x,\theta)\le c_Q b(c,x)$ for all $\theta\in [0,m]$;
\item $b(c,\int_0^T \int_0^m a q(y)(1+y)\alpha_s(dy)ds)\le c_{\text{flow}}b(c,x)$ for all $\alpha\in \cR$. 
\end{enumerate}
Here $\cR$ is the space of {\em relaxed policies}, i.e. of measurable maps
$[0,\infty)\rightarrow  \cP([0,m])$, where $\cP([0,m])$ is the space
of all probability measures on the Borel $\sigma$-algebra on $[0,m]$.
See again \cite{BR2010}.
\end{definition}

\begin{definition}\label{def:gamma}
Let $b:[0,\infty)\times \R\longrightarrow [0,\infty)$ be a measurable function.
For some fixed $\gamma>0$ we define
\[
b(c,x,t):=b(c,x)\exp(\gamma(T-t))\,.
\]
Further we define, for any measurable function 
$v:[0,\infty)\times \R \times [0,T]\longrightarrow \R$,
\[
\|v\|_b:=\underset{(c,x,t)}{\mathrm{ess\;sup}}\frac{|v(c,x,t)|}{b(c,x,t)},
\]
where we set $\frac{0}{0}:=0$,
and we denote by $\B_b$ the Banach space
\[
\B_b:=\{v:[0,\infty)\times \R\times [0,T]\rightarrow [0,\infty): v \text{ is measurable and }\|v\|_b<\infty\}\,.
\]
\end{definition}

\begin{theorem}[Verification Theorem]\label{th:verification}
Let a piecewise deterministic Markov decision process be given with 
a bounding function $b$,  and 
$\E\big(|b(C_T,X^\theta_T)|\big|C_0=c,X_0=x\big)<\infty$ for all 
$\theta,x$. Suppose that 
$v\in C^{0,1,1}([0,\infty)\times\R\times[0,T])\cap \B_b$
is a solution of the HJB equation and that $f^*$ is a maximizer of the
HJB equation and defines a state process $(X^*_t)$.

Then $v=V$ and $\theta^*=f^*(X^*_{t-})$ is an optimal Markov policy 
(in feedback form).
\end{theorem}

\begin{proof}
This can be proved like \cite[Theorem 8.2.8]{BR2011}.
\end{proof}

\begin{remark}
In the statement of Theorem 8.2.8 in \cite{BR2011} there is 
another condition required, namely that $\alpha_b<1$ for a
number $\alpha_b$ depending on $b,Q$ and the arbitrary $\gamma$ from
Definition \ref{def:gamma}. But it is shown in \cite{BR2010} that
for finite horizon problems $\gamma$ can always be chosen large enough to 
satisfy $\alpha_b<1$.
\end{remark}

For the sake of brevity we fix $k$ for the remainder of this 
subsection and we
suppress the dependence of $V$ on $k$.

The generator of $(C_t,X^\theta_t)$ is 
\begin{align*}
\cA^\theta(v)(c,x)
&=q(\theta)\lambda\E (v(c+Y,x-Y)-v(c,x))\\
&\quad+(M-q(\theta))\lambda\E(v(c+Y,x)-v(c,x))\,,
\end{align*}
for $v:[0,\infty)\times\R\longrightarrow\R$ bounded and measurable.
With this, the HJB equation for our problem is
\begin{equation}
0 =\sup_\theta\Big(V_t(c,x,t)
+V_x(c,x,t)a(1+\theta_t)q(\theta_t)
+\cA V(c,x,t)\Big)
\end{equation}
We introduce the shorthand notations
\beast
\hat{V}(c,x,s)&:=&\lambda \E\big(V(c+Y,x,s)-V(c,x,s)\big) \\
\bar{V}(c,x,s)&:=&\lambda \E\big(V(c+Y,x-Y,s)-V(c+Y,x,s)\big) 
\eeast

Reordering of terms 
gives the compact form
\begin{equation}
\begin{array}{rcl}
V_t+M\hat{V}+\sup_{\alpha\in[0,m]}[q(\alpha)(a(1+\alpha)V_x+\bar{V})]&=&0\\[0.4em]
V(c,x,T)&=&u(x+k\psi(c))\,.
\end{array}
\label{eq:HJBcat0}
\end{equation}
We note that if $V\in C^{0,1,1}([0,\infty)\times\R\times[0,T])\cap \B_b$, then $V_x>0$ since $u$ is strictly increasing and
concave, so that the HJB equation can be written as
\begin{equation}
\begin{array}{rcl}
V_t+M\hat{V}+V_x\sup_{\alpha\in[0,m]}[q(\alpha)(a(1+\alpha)+\frac{\bar{V}}{V_x})]&=&0\\[0.4em]
V(c,x,T)&=&u(x+k\psi(c))\,.
\end{array}
\label{eq:HJBcat}
\end{equation}

For all $z\in \R$ 
the function $\alpha\mapsto q(\alpha)(a(1+\alpha)+b)$ is continuous on $[0,m]$
and therefore attains its maximum.
Define the function
\be\label{eq:defmu}
\mu{(z)}:=\max{\{q(\alpha)(a(1+\alpha)+z):{\alpha}\in{[0,m]}\}}\,.
\ee
With this we can write down the following backward equation in $V$,
\be
{\label{eq:beV}}
\begin{array}{rcl}
V_t(c,x,t)+M\hat{V}(c,x,t)+V_x(c,x,t)\mu\left(\frac{\bar{V}(c,x,t)}{V_x(c,x,t)}\right)&=&0\\[0.4em]
V(c,x,T)&=&u(x+k\psi(c))\,.
\end{array}
\ee

%As an example one can compute $\mu$ in the case of linear demand.
%We will show in Proposition \ref{th:lin_max} that for
%$$
%q(\theta)=\left\{
%\begin{array}{c@{\;\text{ if }\;}c}
%M&\theta<0\\
%M(1-\theta/m)&0\le \theta \le m\\
%0& \theta>m
%\end{array}
%\right.
%$$
%we get 
%$$
%\mu(z)=\left\{
%\begin{array}{cl}
%0&z\le-a(m+1)\\
%M(a+z)&z\ge a(m-1)\\
%\frac{M}{4am}(a(1+m)+z)^2&\mbox{else}\,.
%\end{array}
%\right.
%$$
% 
\begin{assumption}
From now on, we restrict our considerations to exponential utility,  
$u(x)=-e^{-\eta x}$. 
\end{assumption}

We make the usual ansatz $V(c,x,t)=u(x)e^{-\eta W(c,t)}$, such that
\beast
V_t(c,x,t)&=&-\eta W_t(c,t)V(c,x,t)\\
V_x(c,x,t)&=&-\eta V(c,x,t)\\
\hat V(c,x,t)
&=&V(c,x,t)\lambda\E(e^{-\eta(W(c+Y,t)-W(c,t))}-1)\\
\bar V(c,x,t)
&=&V(c,x,t)\lambda\E((e^{\eta Y}-1)e^{-\eta(W(c+Y,t)-W(c,t))})\,.
\eeast

We  introduce similar short-hand notations as before:
\beast
\hat W(c,t)&:=&-\frac{1}{\eta}\lambda\E(e^{-\eta(W(c+Y,t)-W(c,t))}-1)\\
\bar W(c,t)&:=&
-\frac{1}{\eta}\lambda\E((e^{\eta Y}-1)e^{-\eta(W(c+Y,t)-W(c,t))})\,.
\eeast
%\begin{lemma}\label{th:Wle0}
%If $Y\ge 0$, then $\bar W(c,t)\le 0$ with strict
%inequality if $\P(Y>0)>0$. 
%\end{lemma}

Substituting into the backward equation \eqref{eq:beV} for $V$
and dividing by $-\eta V$ gives us a backward equation for $W$:
\be
\begin{array}{rcl}
W_t(c,t)+M\hat{W}(c,t)+\mu{(\bar{W}(c,t))}&=&0\\[0.4em]
W(c,T)&=&k\psi(c)
\end{array}\label{eq:beW}
\ee

Note that since $Y\ge 0$ we have $e^{\eta Y}-1\ge 0$ so that $\bar W(c,t)\le 0$
always, with strict inequality  if $\P(Y>0)>0$.
As a consequence, the argument of
$\mu$ will always be negative.

Note further that in order to have finite-valued $\bar W$
we need to make the following assumption.

\begin{assumption}
We assume $\E(e^{\eta Y})<\infty\,$.
\end{assumption}

If we can show that the backward equation \eqref{eq:beW} has a bounded and 
continuous solution
which is differentiable with respect to the time component, and if we can
present a maximizer for this solution, then we
are done. We defer those proofs to Section \ref{sec:verification}.

For the time being we assume the existence of $W$ and $V$, 
and see what we can do with it.

\subsection{Utility indifference price}\label{sec:buyer-price}

We can now -- provided we can solve the corresponding backward equation (\ref{eq:beW}) --  compute the utility indifference
price of a derivative with continuous and bounded payoff $\psi$.
At time $t$ the maximum expected terminal utility of terminal wealth if the
derivative is bought at price $p$ is 
given by
\beast
V(x-p,c,t,k)&=&u(x-p)\exp(-\eta W(c,t,k))\\
&=&u(x)\exp(\eta p)\exp(-\eta W(c,t,k))\,,
\eeast
where $W(.,.,k)$ is the solution to the backward equation (\ref{eq:beW}).
With no derivative bought, the maximum expected terminal utility of 
terminal wealth at time $t$ is given by
$$
V(c,x,t)=u(x)\exp(-\eta W(c,t,0))\,.
$$

Therefore equation (\ref{eq:buyers-price}) takes on the following simple form:
\be\label{eq:indiff_price}
p=p(c,t,k)=W(c,t,k)-W(c,t,0)\,.
\ee
It should be noted that, since $W(c,T,0)\equiv 0$, 
the  expression $W(c,T,0)$  does not depend on $c$. 
It follows from the equation (\ref{eq:beW})
for $W$ that
also $W(c,t,0)$ for $t<T$ does not depend on $c$. 
$W(.,.,0)$ therefore may be computed
as the solution of a one-dimensional ordinary differential equation: Let $W^0$
denote the solution to
\be\label{eq:bew0}
\begin{array}{rcl}
W^0_t(t)+M\hat{W}^0(t)+\mu{(\bar{W}^0(t))}&=&0\\[0.4em]
W^0(T)&=&0\,.
\end{array}
\ee
Then $W(c,t,0)=W^0(t)$ for all $c,t$.  

\begin{lemma}\label{th:constant}
$W^0(t)= \mu(-\frac{1}{\eta}\lambda \E(e^{\eta Y}-1)) (T-t)$.
\end{lemma}

\bproof
Indeed, we see that if $W^0(t)=\kappa(T-t)$ for some constant $\kappa$, 
we have
$\hat W^0(t)=0$ and 
$\bar W^0(t)=\frac{1}{\eta}\lambda \E(e^{\eta Y}-1)$ for all 
$t\le T$. Therefore, $W^0(t)$ solves
$$
\begin{array}{rcl}
W^0_t(t)+M\hat{W}^0(t)+\mu{(\bar{W}^0(t))}&=&0\\[0.4em]
W^0(T)&=&0
\end{array}
$$
iff $\kappa=\mu(-\frac{1}{\eta}\lambda \E(e^{\eta Y}-1))$.

The uniqueness of this solution follows from the general existence and 
uniqueness theorem in Section \ref{sec:verification}, Theorem \ref{th:backward_ode}.
\eproof

Using this lemma we can write down 
a backward equation for $p$: Since  from equation (\ref{eq:indiff_price})
and Lemma \ref{th:constant}
\be
p=p(c,t,k)=W(c,t,k)-\kappa(T-t)
\ee
with $\kappa=\mu(-\frac{1}{\eta}\lambda \E(e^{\eta Y}-1))$,
we get from (\ref{eq:beW})
\be\label{eq:bep}
\begin{array}{rcl}
p_t(c,t,k)-\kappa +M\hat{p}(c,t,k)+\mu{(\bar{p}(c,t,k))}&=&0\\[0.4em]
p(c,T,k)&=&k\psi(c)\,,
\end{array}
\ee
where 
\be\label{eq:bep-abbrev}
\begin{array}{rcl}
\hat p(c,t,k)&:=&-\frac{1}{\eta}\lambda\E(e^{-\eta(p(t,c+Y,k)-p(c,t,k))}-1)\\[0.4em]
\bar p&:=&-\frac{1}{\eta}\lambda\E((e^{\eta Y}-1)e^{-\eta(p(t,c+Y,k)-p(c,t,k))})\,.
\end{array}
\ee

We now look at another aspect of the backward equation in the special case
where
the payoff is such that for some positive $L$ we have 
$\psi(c)=A$ for all $c>L$. This is certainly satisfied for the 
aforementioned examples of spread option and CAT-bond and indeed for 
most reasonable bounded payoff functions.

Under this assumption we have for $c\ge L$ that $\hat p(c,t)=0$ and 
$\bar p(c,t)=-\frac{1}{\eta}\lambda \E(e^{\eta Y}-1)$ for $c>L$.
Therefore we get from \eqref{eq:bep} and the definition of $\kappa$ 
$$
p_t(c,t)=0
$$
for $c\ge L$, and since $p(c,T)=A$, 
$$
p(c,t)=A\,
$$
for $c\ge L$.
This means that if $\psi$ is constant above a cutoff level $L$, then 
we have to compute $W$ for $c\in [0,L]$ only. This obviously simplifies the
numerics. We want to stress however, that all the theoretical results hold
without the above assumption and the numerics can deal with this case quite
analog to many other pricing models where the solution on an unbounded
interval is approximated by a function on a compact interval.

\subsection{Linear demand}\label{sec:linear-demand}

We now consider the particularly simple 
case where $q$ is linear
on the interval $[0,m]$, i.e.
$$
q(\theta)=M\min(1,\max(1-\frac{\theta}{m},0))\,.
$$
Here $\mu$ and  can simply be calculated:
$$
\mu(z)=\left\{
\begin{array}{cl}
0&z\le-a(m+1)\\
M(a+z)&z\ge a(m-1)\\
\frac{M}{4am}(a(1+m)+z)^2&\mbox{else}\\
\end{array}
\right.
$$
and the maximum is attained in 
$$
\gamma(z):=\left\{
\begin{array}{cl}
m&z\le-a(m+1)\\
0&z\ge a(m-1)\\
\frac{a(m-1)-z}{2a}&\mbox{else}\,.
\end{array}
\right.
$$
The simple proof is left to the reader.
%\bproof
%Fix $z\in\R$ and define $f(\alpha)=q(\alpha)(a(1+\alpha)+z)$. 
%For $\alpha<0$ we have $q(\alpha)=M$ such that $f'(\alpha)=Ma>0$.
%So there can never be a maximum strictly to the left of $0$.  
%
%For $\alpha>m$ we have $q(\alpha)=0$, so either all of $[m,\infty)$ is
%a maximum or there is no maximum to the right of $m$.
%
%We have three cases: If there exists $\alpha\in (0,m)$ such that
%$f'(\alpha)=0$, then we have a unique maximum in $\alpha$ since $f$ is
%quadratic with a negative highest-order coefficient. From this it also 
%follows that there cannot be a maximum outside $(0,m)$. 
%
%If we have $f'(\alpha)>0$ for all $\alpha\in(0,m)$, then we obviously have
%a maximum in every point of $[m,\infty)$.
%
%If $f'(\alpha)<0$ for all $\alpha\in(0,m)$, then there is a unique maximum
%in $0$. The remainder of the proof consists of easy calculations.
%\eproof

\subsection{Demand functions with unique optimal risk loading}
\label{sec:non-linear_demand}

In our general derivation the only requirements on the demand function $q$ were
that it be continuous, non-increasing with  
\be\label{eq:q-boundary}
q(\theta)=\left\{
\begin{array}{c@{\;\text{ if }\;}c@{\;\,\!}c@{\;\,\!}l}
M&\theta&\le& 0\\
0&\theta&\ge& m\,.
\end{array}\right.
\ee
In general such a function $q$ will lead to more than one  optimal strategy.
Though the value function does not depend on the particular choice of the
optimal strategy, multiple optimal strategies generate practical difficulties,
for example for numerical computation of the optimal strategy/value functions.

Linear demand is obviously not the only example that allows for exact and
unique computation of the optimal strategy.  Let us consider the class of
functions $q$ which satisfy \eqref{eq:q-boundary}, are twice continuously
differentiable on $(0,m)$, have negative derivative on  $(0,m)$ and for which 
$$
f_q(\alpha,z):=q(\alpha)(a(1+\alpha)+z)
$$
has a unique maximum in $(-\infty,m]$ for all $b$. 

From our assumptions on $q$ we have
$q(\alpha)\ge 0$ so that $f_q(\alpha,z)<0$ iff $a(1+\alpha)+z<0$.
Since $f_q(m,z)=0$, we therefore know that the optimal $\alpha$, if it exists,
must satisfy $a(1+\alpha)+z>0$, i.e. $\alpha>-1-\frac{z}{a}$.

We are therefore interested in demand functions $q$ for which 
$f_q(.,z)$ has a unique maximum in 
$(-1-\frac{z}{a},m]$  if $-1-\frac{z}{a}<m$. 
(For $-1-\frac{z}{a}\ge m$ the function $f_q(.,z)$ attains its maximum in 
$m$.)

A sufficient condition for this is that 
$\alpha \mapsto f_q(\alpha,z)$ is strictly concave on $[-1-\frac{z}{a},m]$.

\begin{theorem}\label{th:special-demand}
Let $q$ be of the form
$$
q(\theta)=\left\{
\begin{array}{cccc@{\;\,\!}l}
M&\text{ if }&\theta&\le& 0\\
0&\text{ if }&\theta&\ge& m\\
M-\int_0^\alpha e^{-\frac{2\xi}{1+m}}H(\xi)d\xi&else
\end{array}\right.
$$
for some function $H:[0,m]\rightarrow \R$ satisfying
\begin{enumerate}
\renewcommand{\theenumi}{\roman{enumi}}
\item $H$ is differentiable on $(0,m)$;
\item $H'>0$ on $(0,m)$;
\item $H>0$ on $(0,m)$;
\item $\int_0^m e^{-\frac{2\xi}{1+m}}H(\xi)d\xi=M$.
\end{enumerate}
Then $f_q(.,z)$ is strictly concave on $[-1-\frac{z}{a},m]$.
\end{theorem}

\begin{proof}%[Proof of Theorem \ref{th:special-demand}]
$$
\frac{\partial^2}{\partial\alpha^2}f_q(\alpha,z)
=q''(\alpha)(a(1+\alpha)+z)+2 q'(\alpha)a\,,
$$
which is negative for $\alpha\in (-1-\frac{z}{a},m)$ iff
$$
\frac{q''(\alpha)}{q'(\alpha)}>-\frac{2a}{a(1+\alpha)+z}\,.
$$
Since $z\le 0$, the right hand side is always greater or equal to 
$-\frac{2a}{a(1+\alpha)}=-\frac{2}{(1+\alpha)}\ge -\frac{2}{(1+m)}$.

We have therefore shown that if 
\be\label{eq:concave1}
\frac{q''(\alpha)}{q'(\alpha)}=-\frac{2}{1+m}+h_1(\alpha)
\ee
for some continuous function $h_1:[0,m]\rightarrow\R$ which is
positive on $(0,m)$, then
$f''_q(\alpha,z)<0$ for all $\alpha\in (-1-\frac{z}{a},m)$, $z\le 0$.
But 
$$
\begin{array}{crcl}
&\frac{q''(\alpha)}{q'(\alpha)}&=&-\frac{2}{1+m}+h_1(\alpha)\\
\Longleftrightarrow& \frac{d}{d\alpha}\log(|q'(\alpha)|)&=&-\frac{2}{1+m}+h_1(\alpha)\\
\Longleftrightarrow& \log(|q'(\alpha)|)&=&-\frac{2\alpha}{1+m}+h_2(\alpha)\\
\Longleftrightarrow& q'(\alpha)&=&-\exp\left(-\frac{2\alpha}{1+m}+h_2(\alpha)\right)\\
\Longleftrightarrow& q'(\alpha)&=&-\exp\left(-\frac{2\alpha}{1+m}\right)H(\alpha)
\end{array}
$$
where $h_2$ is a primitive function of $h_1$, i.e. $h_2'=h_1$ and 
$H(\alpha)=\exp(h_2(\alpha))$. $H$ is a positive, continuously differentiable
function with  $H'>0$.
So if with this $H$ we define
$$
q(\alpha)=M-\int_0^\alpha\exp\left(-\frac{2\xi}{1+m}\right)H(\xi)d\xi
$$
we have $q(0)=M$ and if we further have
$$
\int_0^m\exp\left(-\frac{2\xi}{1+m}\right)H(\xi)d\xi=M\,,
$$
then $q(m)=0$.
\eproof

Note that we recover linear demand for 
$H(\xi)=\frac{M}{m}e^{\frac{2\xi}{1+m}}$.
Other simple examples are arrived at by using a function $H$ of the form
$H(\xi)=c\, P(\xi) e^{\frac{2\xi}{1+m}}$, where $P$ is some
polynomial which satisfies $P(\xi)>0$ and $P'(\xi)>0$ for $\xi\in(0,m)$ and 
$$
c=M \left(\int_0^m P(\xi)d\xi\right)^{-1}\,.
$$
Other noteworthy examples are provided by
$q(\alpha)=M(1-(\frac{\alpha}{m})^\nu)$ on $[0,m]$ with $\nu>0$.

\subsection{Numerical example}\label{sec:numerical-example}

In this section we consider a numerical example for the model
proposed earlier. All further illustrations will refer to this setup.

We concentrate on the special case of linear demand from 
Section \ref{sec:linear-demand} and take
the following values for the problem: $T=\frac{1}{4}$, $\lambda=0.01$,
$M=10^4$, $m=2$.
The claim sizes $Y_i$ are distributed on $\{\delta,2 \delta,3 \delta,4 \delta, 5\delta\}$ with
$\delta=10^5$ and corresponding probabilities  
$\frac{1}{8},\frac{3}{8},\frac{2}{8},\frac{1}{8},\frac{1}{8}$. 
The risk aversion coefficient is $\eta=10^{-6}$.

We consider the payoff $\psi(c)=\max(0,\min(c-K,L-K))$, where
$K=10^7$, $L=3 \cdot 10^7$. We therefore have 
$\psi(c)=A$ for $c>L$ with $A=2 \cdot 10^7$.
Note that in our setup  the  PIDEs \eqref{eq:beW} and  
\eqref{eq:risk-neutral} become ordinary differential equations in $\R^n$,
where $n=\frac{A}{\delta}+1$. This has the consequence that this example
can be computed rather efficiently.

Figure \ref{fg:price} shows the price $p^b(.,t)$ of the derivative as a function
of $C_t$. The darkest line shows the price at expiry, which is equal to the
payoff, the lightest line shows the price at time $t=0$. Both axes 
are million units of currency.

\begin{figure}
\begin{center}
\ifpdf
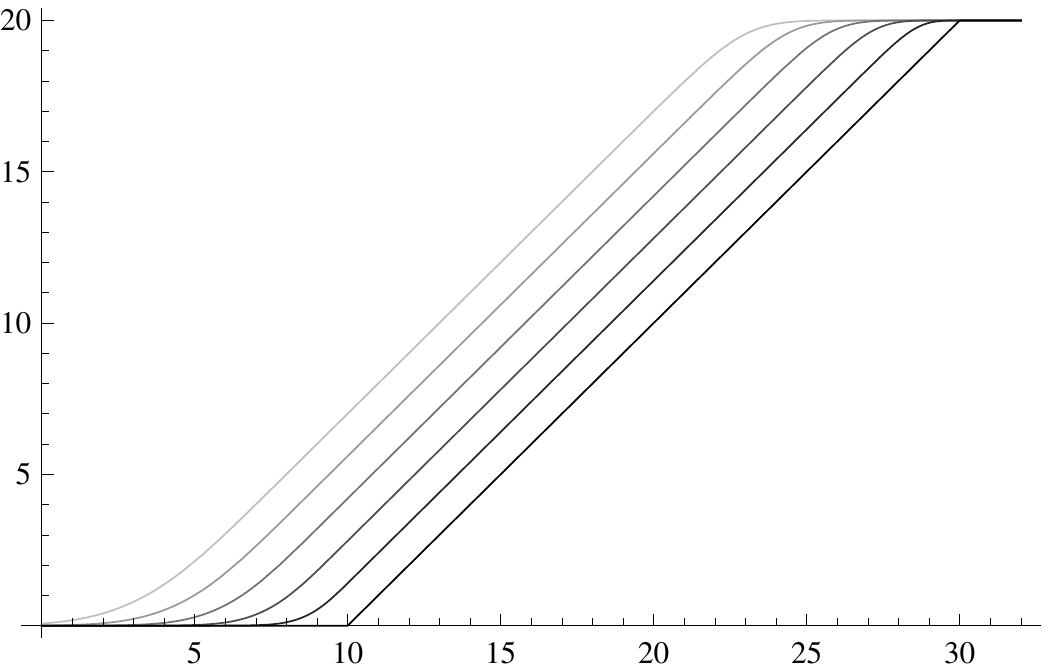
\else
\input{gfx_uibp.pstex_t}
\fi
\end{center}
\captionof{figure}{The buyer's price $p^b(.,t)$ for various values of $t\in [0,T]$. Lighter shades of gray correspond to earlier times.}\label{fg:price}
\end{figure}

As is to be expected, the price of the derivative is a smoothed version
of the payoff shifted to the left. This shift is due to the fact 
that the claims process is non-decreasing with time.

Figure \ref{fg:riskloading} shows the risk loading $\theta_t$ as a function
of $C_t$. The darkest line shows the risk-loading close to expiry, 
the lightest line shows the risk-loading at time $t=0$. 
The $x$-axis is in million units of currency.
We see that for some parameters (e.g. $t=0$, $C_t=15\cdot 10^6$), due to the 
presence of the derivative the risk loading is pushed down roughly 
from 1.09 to 0.93. That means that the derivative makes the insurance more
than 10 percent cheaper.  Obviously this effect vanishes for $C_t>L$, in which
case the derivative corresponds to a deterministic payment, such that 
the risk loading reverts to the one without a derivative present.

For $C_t$ far below the ``strike'' $K$ the risk loading is also higher, which
can be explained by the relatively high probability for the derivative
to have zero payoff, such that the risk loading is close to that without
a derivative present.

\begin{figure}
\begin{center}
\ifpdf
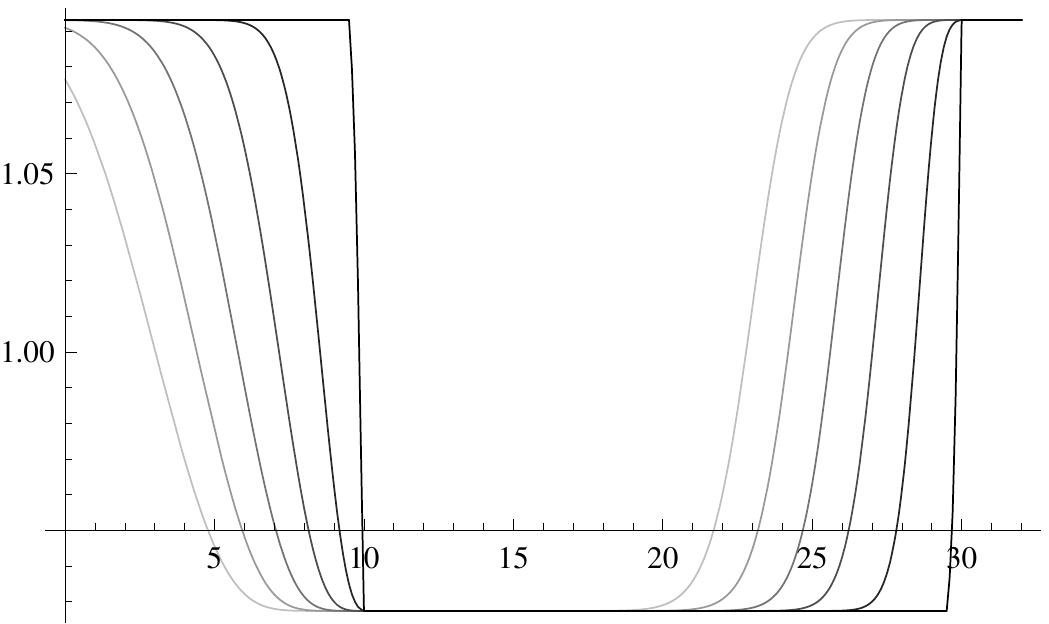
\else
\input{gfx_theta.pstex_t}
\fi
\end{center}
\captionof{figure}{Here we plot the risk loading $\theta(c,t)$ as a function of $c$. Lighter shades of gray correspond to earlier times. }\label{fg:riskloading}
\end{figure}

\subsection{Verification}\label{sec:verification}

Now we want to convince ourselves that
there exists a solution to (\ref{eq:beW}), 
with a corresponding maximizer and 
that the conditions of Theorem 
\ref{th:verification} are satisfied.

\begin{lemma}
$b(c,x):=\exp(\eta x)$ is a bounding function for our piecewise deterministic Markov decision model
with bounded payoff function $\psi$.
\end{lemma}

\begin{proof}
(i) $u(c,x)=-\exp(-\eta(x+\psi(c))$ such that 
$|u(c,x)|=\exp(-\eta(x+\psi(c))\le \exp(\eta \|\psi\|_\infty)b(c,x)$.

\noindent (ii) 
\begin{align*}
\int b(c,x)Q(dc\times dx|c,x,\theta)
&=\int \exp(\eta x)Q(dc\times dx|c,x,\theta)\\
&=q(\theta)\lambda\E(\exp(\eta (x-Y)))+(M-q(\theta))\lambda\E(\exp(\eta x)))\\
&\le b(c,x) \sup_\theta \left(q(\theta)\lambda\E(\exp(-\eta Y))+(M-q(\theta))\lambda\right)\\
&\le \lambda M b(c,x)\,.
\end{align*}
\noindent (iii) Let $\zeta:=\sup_{y\in [0,m]}{a q(y)(1+y)}$. Then
\begin{align*}
b(c,x+\int_0^T\int_0^m a q(y)(1+y)\alpha_s(dy)ds)
&=\exp\left(\eta x+\eta \int_0^T\int_0^m a q(y)(1+y)\alpha_s(dy)ds\right)\\
&\le\exp(\eta T \zeta)b(c,x)\,.
\end{align*}
\end{proof}

Recall the function
$$
\mu{(z)}:=\max{\{q(\alpha)(a(1+\alpha)+z):{\alpha}\in{[0,m]}\}}
$$ 
and define the multivalued correspondence 
$$
\Gamma(z):=\{\alpha\in[0,m]:q(\alpha)(a(1+\alpha)+z)=\mu(z)\}\,.
$$
\begin{lemma}\label{th:mu_lipschitz}
\begin{enumerate}
\item $\mu$ is a convex function on $\R$.
\item $\mu$ is Lipschitz-continuous on compact sub-intervals of $\R$.
\end{enumerate}
\end{lemma}

\begin{proof}%[Proof of Lemma \ref{th:mu_lipschitz}]
From our assumptions on $q$ it follows that there exists a continuous inverse
$q^{-1}$ to $q$ on $[0,M]$. 
First note that 
\beast
\mu{(z)}
&=&\max{\{q(\alpha)(a(1+\alpha)+z):{\alpha}\in{[0,m]}\}}\\
&=&a \max{\{q(\alpha)\alpha+q(\alpha)(z/a+1):{\alpha}\in{[0,m]}\}}\\
&=&a \max{\{\beta q^{-1}(\beta)+\beta(z/a+1):{\beta}\in{[0,M]}\}}\\
&=&a \max{\{-f(\beta)+\beta(z/a+1):{\beta}\in{[0,M]}\}}\\
&=&a f^*(z/a+1)\,,
\eeast
where 
$$
f(\beta):=\left\{\begin{array}{cl}
-\beta q^{-1}(\beta)& \beta\in[0,M]\\
\infty&\beta\in\R\backslash[0,M]
\end{array}\right.
$$
and
$f^*$ denotes the convex conjugate of $f$, cf. §12 in \cite{rock}.
$f^*$ is a convex function on all of $\R$, see  Theorem 12.2 in \cite{rock}.
$f^*$ is real-valued on all of $\R$ since
$$
\beta \beta^*-f(\beta)=\left\{\begin{array}{cl}
\beta\beta^*+\beta q^{-1}(\beta)& \beta\in[0,M]\\
-\infty&\beta\in\R\backslash[0,M]
\end{array}\right.
$$
such that
$$
\sup_{\beta\in \R}(\beta \beta^*-f(\beta))=\max_{\beta\in[0,M]}(\beta\beta^*+\beta q^{-1}(\beta))\in\R\,.
$$
Therefore $f^*$ is Lipschitz on compact intervals,  cf. Theorem 10.4 in \cite{rock},
and so is $\mu$.
\eproof

\begin{theorem}\label{th:backward_ode}
Let $\psi$ be a continuous and bounded function on $\R$ and
let $\E(e^{\eta Y})<\infty$.
Then the backward equation (\ref{eq:beW}) for $W$ has
a unique solution.

Moreover, the solution is bounded.
\end{theorem}

%\bproof See Appendix \ref{sec:main-proofs}.  \eproof
\begin{proof}%[Proof of Theorem \ref{th:backward_ode}]
Consider the Banach space $C_b(\R)$ of bounded continuous functions on $\R$. 
The backward equation (\ref{eq:beW}) is just an initial value problem 
for a $C_b(\R)$-valued function,
\beast
w'(t)&=&G(t,w(t))\\
w(T)&=&\psi\,,
\eeast
where $G(t,w)(c)=-M\hat{w}(c,t)-\mu{(\bar{w}(c,t))}$.
Using Lemma \ref{th:mu_lipschitz} it is readily shown that
$G$ satisfies a local Lipschitz condition in the second variable with
respect to the sup-norm. Therefore the classical Picard-Lindel\"of theorem
on existence and uniqueness of solutions of ODEs gives us the unique solution
$w$ to the initial value problem. Since each $w(.)$ is bounded and $w$ is 
continuous, we have that the function
$$
\begin{array}{ccccc}
[0,T]\times\R &\longrightarrow& \R\\
(c,t)&\longmapsto& w(t)(c)
\end{array}
$$
is bounded.
\eproof

\begin{lemma}\label{th:meas-choice}
There is a measurable function 
$\gamma:\R\longrightarrow[0,m]$ such that 
$$
q(\gamma(z))(a(1+\gamma(z))+z)=\mu(z)
$$ 
for all $z\in \R$.
\end{lemma}

\begin{proof}%[Proof of Lemma \ref{th:meas-choice}]
Recall the correspondence $\Gamma$
which maps $z\in \R$ to the compact set of all $\alpha\in[0,m]$ which maximize
$q(\alpha)(a(1+\alpha)+z)$. Berge's Theorem of the Maximum (\cite{berge} p.116),
states that $\Gamma$ is upper hemi-continuous, such that its graph is
closed. The graph $\cG_\Gamma$ is therefore a closed subset of $\R\times[0,m]$ 
and the projection of $\cG_\Gamma$ to $\R$ is $\R$. From von Neumann's 
Measurable Choice Theorem (see \cite{dixmier} Appendix V) it therefore follows that there exists a
measurable function $\gamma:\R\longrightarrow[0,m]$ such that 
$q(\gamma(z))(a(1+\gamma(z))+z)=\mu(z)$ for all $z\in \R$.
\eproof

Theorem \ref{th:backward_ode}
and Lemma \ref{th:meas-choice} together show that the assumptions of the 
Verification Theorem \ref{th:verification} are satisfied.

\section{Certainty equivalence price and limit prices}\label{sec:limit-prices}

In the preceding section we computed the utility indifference
price. Having done this we might ask whether there can actually
be a trade, that is whether it may occur that $p(c,t,k)=p^b(c,t,k)\ge p^s(c,t,k)=-p(c,t,-k)$.
This does not seem to be the case too often. All of our numerical
examples show that $p^b(k)<p^s(k)$. This is in line with, e.g., the findings 
by Takino \cite{tak}, who computes indifference prices of European claims in 
a stochastic volatility model with partial information.

\subsection{Certainty equivalence price}

Another interesting question is what an investor is willing to charge
for the derivative if she cannot hedge the derivative. 
The hedging strategy we proposed
earlier can only be realized by an insurance company.
It is not unreasonable to assume that the counterpart in such a deal is {\em not}
an insurance company. In that case the utility indifference price
of the buyer coincides with the certainty equivalence price. 

It is also reasonably to assume that on the seller's side the derivative
is split up in the way mentioned above, i.e. each seller sells only the 
$N$-th part of the whole
derivative. The example we have in mind is that of a CAT-bond which is
 denominated into $N$ units.

Let us denote the certainty equivalence price for the seller by
$\pi^s(c,t,k)$, i.e. the solution to
$$
\E\big(-\exp(-\beta y_0)|C_t=c\big)=\E\Big(-\exp\big(-\beta(y_0+\pi^s(c,t,k)-k\psi(C_T))\big)|C_t=c\Big)\,,
$$
where we have assumed exponential utility with coefficient of 
risk aversion $\beta$.
That is
$$
\pi^s(c,t,k)=\frac{1}{\beta}\log\big(\E(\exp(\beta k\psi(C_T))|C_t=c)\big)\,.
$$
$\pi^s$ may be computed either directly using the distribution of $C_T$
or using a backward equation which can be derived similarly to the 
backward equation for $W$:
\be\label{eq:equivalence}
\begin{array}{rcl}
\pi^s_t(c,t,k)+M\frac{\lambda}{\beta} \E\left(e^{\beta\pi^s(c+Y,t,k)}-e^{\beta\pi^s(c,t,k)}\right) &=&0\\
\pi^s(c,T,k)&=&k\psi(c)\,.
\end{array}
\ee

Therefore the derivative can only be sold in denomination $N$, between
$N$ sellers without the opportunity to hedge and a buyer with the opportunity
to hedge, if
$$
p^b(c,t,1)\ge N \pi^s(c,t,1/N)
=\frac{N}{\beta}\log(\E(\exp(\frac{\beta}{N}\psi(C_T))|C_t=c))\,
$$
for some $N$.
For bounded $\psi$ and $N\rightarrow \infty$ we have 
$\exp(\frac{\beta}{N}\psi(C_T))\approx 1+\frac{\beta}{N}\psi(C_T)$ 
and therefore 
$\E(\exp(\frac{\beta}{N}\psi(C_T))|C_t=c)\approx 1+\frac{\beta}{N}\E(\psi(C_T)|C_t=c)$
and further $\frac{N}{\beta}\log\E(\exp(\frac{\beta}{N}\psi(C_T))|C_t=c)\approx\E(\psi(C_T)|C_t=c)$.
That is 
$$
\lim_{N\rightarrow \infty}N \pi^s(c,t,1/N) =\E(\psi(C_T)|C_t=c)\,.
$$
Define
$$
\pi^0(c,t):=\E(\psi(C_T)|C_t=c)\,,
$$
then $v$ can be computed directly or via the backward equation
\be\label{eq:risk-neutral}
\begin{array}{rcl}
v_t(c,t)+M\lambda \E(\pi^0(c+Y,t)-\pi^0(c,t)) &=&0\\
\pi^0(c,T)&=&\psi(c)\,,
\end{array}
\ee
and gives the limit of the buyer's price for $ N\rightarrow\infty$.

In our numerical example  we have 
$$
p^b(c,t,1)> \lim_{N\rightarrow\infty}N\pi^s(c,t,\frac{1}{N})=\pi^0(c,t)\,,
$$
that is a deal could be stricken for sufficiently large $N$, {\em provided that
the sellers of the derivative are not able to hedge the derivative. }
The corresponding difference between utility indifference buyer's price
and $\pi^0$, the denomination limit of $\pi^s$, is shown in Figure \ref{fg:diff}.

\begin{figure}
\begin{center}
\ifpdf
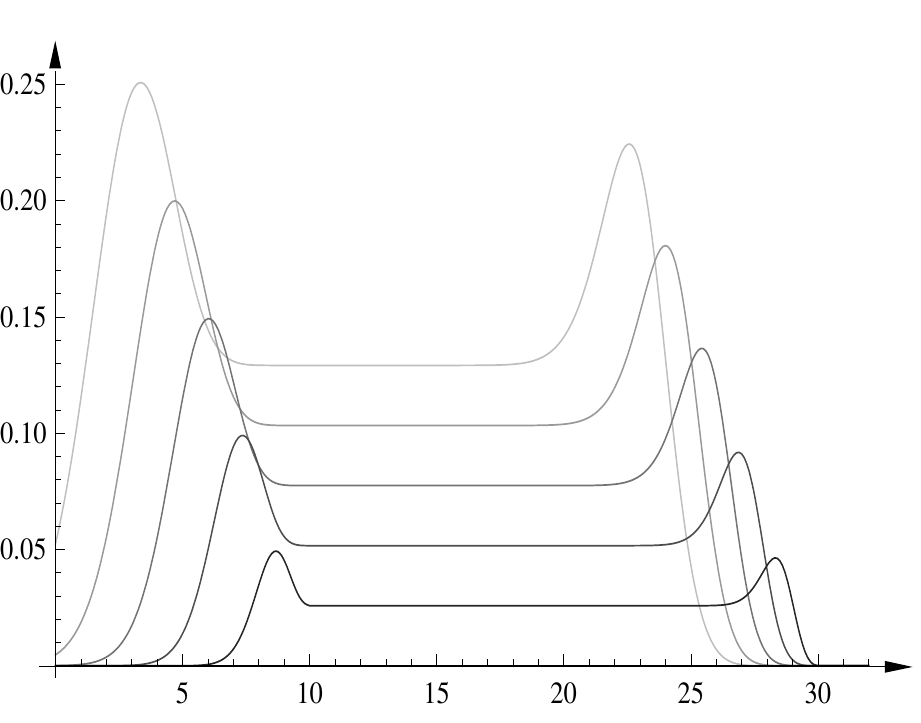
\else
\input{gfx_uibp-rnsp.pstex_t}
\fi
\end{center}
\captionof{figure}{Here we plot $p^b-\pi^0$ as a function of $c$. 
Lighter shades of gray correspond to earlier times. Note that the difference is 
always non-negative, which implies that the derivative can actually be 
traded.}\label{fg:diff}
\end{figure}

\subsection{Risk-neutral limit}

Another interesting quantity is the risk-neutral limit of the indifference
price, that is 
$$
p_0(c,t,k):=\lim_{\eta\rightarrow 0} p_\eta(c,t,k)\,,
$$
where $p_\eta$ is the utility indifference price corresponding to 
risk aversion $\eta$. The risk-neutral limit price has been considered,
for example, in \cite{RoEK,DGRSSS,Bec,becherer}, and it is of some interest 
in that
it gives a linear pricing rule which nevertheless is related to the 
non-linear utility indifference pricing rule.\\

It is not hard to see that our optimization problem  
\eqref{eq:value_f} is meaningful even for $\eta=0$, that is if one
takes linear utility $u(x)=x$. It turns out
that  $p_0(c,t,k)=k \pi^0(c,t)$ where $\pi^0$ is the same as in the preceding 
subsection.

\section{Conclusion and open questions}\label{sec:conclusions}

In our study,  we have modeled the industrial loss index  by a compound Poisson
process and  showed that the insurer can control her wealth process by
adjusting her portfolio via choosing the risk loading.  Our study contributes
to the insurance theory by showing that by issuing CAT bonds and offering
catastrophe coverage  the net expected income of the insurance company 
remains the same while the insurer can lower the premium charge.

 The study has a greater
significance to low-income countries where natural disasters often exceeds the
resources of internal and external sources of relief funding: Using our
strategy, the insurance company in the low-income country can sell CAT bond to
some (ethical) investors and offer affordable insurance services against
risk of low-probability, high-loss  events to the needy/poor vulnerable
customers.

We have discussed to role of the ability to hedge. We have found that in the
natural situation where the seller (or the sellers) of the derivative is
not an insurance company and therefore cannot hedge the derivative via
her portfolio, then the derivative can actually be bought. \\
 
For future research, the following extensions or generalizations of the problem are of interest: 
\begin{itemize}
\item $N$ could be a doubly stochastic process, such that its intensity varies
over time. This would allow for a more realistic modeling of catastrophic
events: 
One could have ``normal'' times, where claims arrive at a low rate and
``catastrophe'' times, where claims arrive at a very high rate. One would 
then probably restrict the policies of the insurance company in a way that does
not permit changing the risk loading during catastrophe times.  
\item Alternatively, one could model catastrophes as events where several
claims happen at the same time and where the insurance get to pay a random
number of claims according to their fraction of the total portfolio. 
\item It would be interesting to allow for some lag for the adjustment of the
demand to a changed risk loading.
\item The $Y_k$'s and $U_k$'s could be made dependent. More specifically,
the parameters of $Y_k$ could be a function of $U_k$. This is reasonable, 
since the different clients are likely to have different claim distributions.
In that setup, the policies of the insurance company would be more complicated
objects, the risk loading would also depend on the parameters of the claim
distributions.
\end{itemize}

%\cite{becherer,zhuzhao,Jaimungal2,barrieu,charpentier}

%\bibliography{mybibliography}

\end{document}